\documentclass[conference,letter]{IEEEtran}

\usepackage[utf8]{inputenc} 
\usepackage[T1]{fontenc}
\usepackage{url}
\usepackage{ifthen}

\usepackage[cmex10]{amsmath} 

\usepackage{amssymb}
\usepackage{url}
\usepackage{array}
\usepackage{psfrag}
\usepackage{pgfplots}
\usepackage[ruled,linesnumbered,vlined,titlenumbered]{algorithm2e}
\usepackage[margin=0.67in,top=0.545in]{geometry}
\usepackage{cite}

\usepackage{amsthm}
\newtheorem{theorem}{Theorem}

\newtheorem{lemma}{Lemma}
\newtheorem{definition}{Definition}
\newtheorem{example}{Example}

\usepackage{tikz}

\newcolumntype{C}{>{$}c<{$}} 

\newcommand{\secref}[1]{Section~\ref{#1}}
\newcommand{\algoref}[1]{Algorithm~\ref{#1}}
\newcommand{\lemref}[1]{Lemma~\ref{#1}}
\newcommand{\theoref}[1]{Theorem~\ref{#1}}
\newcommand{\defref}[1]{Definition~\ref{#1}}
\newcommand{\figref}[1]{Figure~\ref{#1}}
\newcommand{\stepref}[1]{Step~\ref{#1}}
\newcommand{\tabref}[1]{Table~\ref{#1}}

\newcommand{\RS}[2]{\ensuremath{\mathcal{RS}(#1,#2)}}
\newcommand{\LRC}[4]{\ensuremath{\mathcal{LRC}(#1,#2,#3,#4)}}

\newcommand{\LRCnp}{\ensuremath{\mathcal{LRC}}}
\newcommand{\code}{\ensuremath{\mathcal{C}}}
\newcommand{\ceil}[1]{\ensuremath{\left\lceil #1 \right\rceil}}
\newcommand{\floor}[1]{\ensuremath{\left\lfloor #1 \right\rfloor}}

\newcommand{\F}[1]{\ensuremath{\mathbb{F}_{#1}}}

\newcommand{\uv}[1]{\ensuremath{\mathbf{u}_{#1}}}
\newcommand{\dt}[2]{\ensuremath{\text{d}_{\text{H}} (#1 , #2) }}

\newcommand{\nlc}{\ensuremath{n_{l}}}

\IEEEoverridecommandlockouts

\begin{document}
	
\title{List Decoding of Locally Repairable Codes} 

 \author{%
   \IEEEauthorblockN{Lukas Holzbaur, Antonia Wachter-Zeh\thanks{L. Holzbaur's and A. Wachter-Zeh's work was supported by the Technical
   		University of Munich—Institute for Advanced Study,
   		funded by the German Excellence Initiative and European Union Seventh
   		Framework Programme under
   		Grant Agreement No. 291763 and the German Research Foundation (Deutsche
   		Forschungsgemeinschaft,
   		DFG) unter Grant No. WA3907/1-1.}}
   \IEEEauthorblockA{Insitute for Communications Engineering, Technical University of Munich, Germany\\
                     Email: \{lukas.holzbaur, antonia.wachter-zeh\}@tum.de }
 }

\maketitle

\begin{abstract}
  We show that locally repairable codes (LRCs) can be list decoded efficiently beyond the Johnson radius for a large range of parameters by utilizing the local error correction capabilities. The new decoding radius is derived and the asymptotic behavior is analyzed. We give a general list decoding algorithm for LRCs that achieves this radius along with an explicit realization for a class of LRCs based on Reed-Solomon codes (Tamo-Barg LRCs). Further, a probabilistic algorithm for unique decoding of low complexity is given and its success probability analyzed. 
\end{abstract}

\section{Introduction} \label{sec:introduction	}

\IEEEPARstart{T}{he} ever increasing demand for distributed data storage capacity causes rising interest in coding solutions specifically developed for storage systems. On such a massive scale, unreachable or failed servers are no longer an exception but a regular occurrence and recovery from such events has to be done efficiently. Requiring a subset of participating servers to enable recovery from failures, translates to the coding-theoretical problem referred to as \emph{locality}, where in addition to the distance between all positions of codewords, a distance also has to be guaranteed on subsets of codeword positions. A Singleton-like bound on the achievable distance $d$ was derived in~\cite{Gopalan2011} and generalized in~\cite{Kamath2012}, and constructions achieving it were presented in~\cite{Silberstein2013,Tamo2013,Tamo2014} among others. Naturally, as the distance cannot be higher, the maximum decoding radius for \emph{bounded minimum distance} (BMD) decoding can be achieved by these constructions. However, for different decoding goals it is possible to make further use of the additional local distance. In~\cite{Tamo2013}, it was shown that with an asymptotically diminishing probability of failure, more than $d-1$ erasures can be recovered. 

To our knowledge it is not known how to utilize the additional redundancy coming from the locality when considering (list) decoding of \emph{errors}. A list decoder returns all codewords within a specified distance around the received word. It is known that Reed--Solomon (RS) codes, like all linear codes, can be list decoded up to the Johnson radius \cite{Johnson1962} and an explicit algorithm exists \cite{Guruswami1999}. Though it has been shown that some Reed-Solomon codes can be list decoded beyond this radius~\cite{Rudra2013}, there are no known algorithms to achieve this. Optimal LRCs can be constructed as subcodes of RS codes~\cite{Tamo2014}. In this paper, we show that for a large range of parameters LRCs can be list decoded beyond the Johnson radius while the complexity and list size grow polynomially in the code length, when the number of local repair sets is constant. Further, we give a low complexity probabilistic algorithm and analyze the success probability. Finally, an explicit algorithm for list decoding Tamo--Barg LRCs up to the derived radius is given.

\section{Preliminaries} \label{sec:preliminaries}

\subsection{Notations and Definitions}

Denote by $\F{q}$ a finite extension field of order $q=p^m$, where~$p$ is a prime and $m$ is a positive integer. We write $[n]$ for the set of integers $\{i : 1 \leq i \leq n , i \in \mathbb{Z} \}$.

Let $\code$ be an $[n,k,d]$ code and $H \subset [n]$ be a set of coordinates. Denote by $\code_H$ the code obtained by restricting~$\code$ to the coordinates of $H$.

We define \emph{shortening} of an $[n,k,d]$ code $\code$ in position $i$ by a fixed value $\gamma$ as $\code' = \{c | c\in \code ,  c_i = \gamma \}_{[n] \backslash i}$. 

\subsection{Locally Repairable Codes}

A code is said to have \emph{locality} $r$ if every code symbol can be recovered by accessing the values of at most $r$ other positions, i.e., every code symbol is part of a local code of length $r+1$ and distance $2$. The general case is referred to as $(r,\rho)$ locality.

\begin{definition}[Locality]
	A code $\code$ has $(r,\rho)$ (all-symbol) locality if there exists a partition $\mathcal{H} = \{H_1,H_2,...\}$ of $[n]$, with $H_i \cap H_j = \emptyset \;\; \forall \;\; i,j \in [|\mathcal{H}|], i\neq j$, such that the restriction of the code $\code$ to the coordinates of $H_j$ is a code of length at most $r+\rho-1$ and distance at least $\rho$.
\end{definition}

For local distance $\rho = 2$ a Singleton-like upper bound was shown in \cite{Gopalan2011} and later generalized for $\rho \geq 2$ in \cite{Kamath2012} to
\begin{equation} \label{eq:singletonlrc}
	d \leq n-k+1-\left(\ceil{\frac{k}{r}} -1\right)(\rho-1).
\end{equation}
We call a code that achieves this bound with equality an \emph{optimal} code.

In the following, the restriction $\code_{H_j}$ is referred to as a local code. Only codes with local codes of equal length $n_l = |H_j|$ are considered and we restrict ourselves to codes where $n_l | n$ and $r | k$.
We denote a code of length~$n$, dimension~$k$, locality~$r$ and local distance~$\rho$ by $\LRC{n}{k}{r}{\rho}$.

\section{List Decoding of LRCs} \label{sec:listdecoding}

\subsection{New Decoding Radius}\label{subsec:newdecodingradius}
A code of length~$n$ is called~$(\tau,\ell)$-list decodable if the Hamming sphere of radius~$\tau$ centered at any vector~$v$ of length~$n$ always contains at most~$\ell$ codewords~$c\in  \code$. It is known \cite{Johnson1962} that any code of length~$n$ and distance~$d$ is list decodable up to the Johnson radius 
\begin{equation} \label{eq:johnsonradius}
\tau_J = n-\sqrt{n(n-d)}
\end{equation}
with list size polynomial in~$n$.
We denote the number of list decodable errors, i.e. the largest integer smaller than~$\tau$, by
\begin{align}
t=\ceil{\tau-1},\; \text{where} \;\tau-1 \leq t < \tau \label{eq:definet}.
\end{align}
Generally, it is conjectured that the list size increases exponentially in the code length~$n$ when the radius is at least~\eqref{eq:johnsonradius}. While it is known that there are codes for which the bound is not tight and the list decoding radius exceeds the Johnson radius \cite{Guruswami2006}, the behavior of RS codes is still mostly an open problem \cite{Guruswami2012, Rudra2013}. In the following, we show that the list decoding radius 
of certain LRCs exceed the Johnson radius, i.e., the complexity and list size grow \emph{polynomially} in the length when the number of local repair sets~$\frac{n}{n_l}$ is constant.

\lemref{lemma:sigma} establishes a lower bound~$\sigma$ on the number of locally decodable repair sets as a function of the decoding radius.

\begin{lemma}  \label{lemma:sigma}
	Let~$\code$ be an~$\LRC{n}{k}{r}{\rho}$. Denote by~$\tau_g$ and~$\tau_l$ the any global and local decoding radius and let~$t_g$ and~$t_l$ be defined as in~\eqref{eq:definet}. For a codeword~$c \in \code$ and any word~$w$ with~$\dt{c}{w} \leq t_g$, let~$\mathcal{I}\subseteq \left[\frac{n}{n_l}\right]$ be the set of repair set indices~$i$ with~$\dt{c_{H_i}}{w_{H_i}} \leq t_l, \; \forall \; i \in \mathcal{I}$. The cardinality of~$\mathcal{I}$ is bounded by
	\begin{equation}
	|\mathcal{I}| \geq \sigma = \max\left\{0,\frac{n}{n_l}-\frac{\tau_g}{\tau_l} \right\} .\label{eq:sigmaineq}
	\end{equation}
\end{lemma}
\begin{IEEEproof}
	Trivially, the cardinality of~$\mathcal{I}$ is non-negative. The maximum number of repair sets~$H_j$ with~$\dt{c_{H_j}}{w_{H_j}} > t_l$ such that~$\dt{c}{w} = \sum_{j=1}^{\frac{n}{n_l}} \dt{c_{H_j}}{r_{H_j}} \leq t_g$ is given by~$\floor{\frac{t_g}{t_l+1}}$. Subtracting from the total number of repair sets~$\frac{n}{n_l}$ gives 
	\begin{align*}
	\frac{n}{n_l}-\floor{\frac{t_g}{t_l+1}} \geq \frac{n}{n_l}-\frac{t_g}{t_l+1} > \frac{n}{n_l}-\frac{\tau_g}{\tau_l} .
	\end{align*}
	
	\vspace{-1ex}
\end{IEEEproof}

The following theorem provides our main statement.

\begin{theorem}[List Decoding of LRCs] \label{theo:ListDecodingLRCs}
	Let $\ell_{(n,d,\tau)}$ denote the maximum list size when list decoding an~$[n,k,d]$ code with radius~$\tau$. An~$\LRC{n}{k}{r}{\rho}$ is~$(\tau_g,\ell_g)$-list decodable, with
	\begin{equation}
	\tau_g = \left\{
	\begin{array}{ll}
	\frac{d}{\rho} \cdot \tau_l \, , &\text{if}\;\; \sigma>0 \\
	n-\sqrt{n(n-d)}\, , & \text{else}
	\end{array} \right.   \label{eq:jblrc}
	\end{equation}
	and
	\begin{equation}
		\ell_g \leq \binom{\frac{n}{n_l}}{\sigma} \ell_{(n_l,\rho,\tau_l)}^\sigma \ell_{(n-\sigma n_l,d,\tau_g)} ,
	\end{equation}
	where $\tau_l$ is the Johnson radius of the local codes.
\end{theorem}
\begin{IEEEproof}
	By \lemref{lemma:sigma}~$\dt{c_{H_i}}{w_{H_i}} \leq t_l$ holds for at least~$\sigma$ repair sets. These repair sets can be decoded locally and the code can be shortened by these~$\sigma n_l$ positions to an~$(n-\sigma n_l,k-\sigma n_l,d)$ code. The Johnson radius of this code is given by the largest~$\tau_g$ that fulfills
	\begin{align}
	0 &< (n-\sigma n_l-\tau_g)^2 -(n-\sigma n_l) (n-\sigma n_l -d)  \label{eq:tausiglem}\\
	&=(n-\tau_g)^2-n(n-d)+\sigma n_l(2\tau_g-d)  .\label{eq:tauflem} 
	\end{align}
	This is an increasing function in~$\sigma$ as long as~$2 \tau_g \geq d$ (i.e., when BMD decoding is not possible). With~\eqref{eq:sigmaineq} for~$\sigma >0$ it follows that any~$\tau_g \geq d/2$ that fulfills
	\begin{align}
	0&< \left(n-\left(\frac{n}{\nlc} - \frac{\tau_g}{\tau_l}\right) n_l-\tau_g\right)^2 \nonumber \\
	\quad &-\left(n-\left(\frac{n}{\nlc} - \frac{\tau_g}{\tau_l}\right) n_l\right) \left(n-\left(\frac{n}{\nlc} - \frac{\tau_g}{\tau_l}\right) n_l -d\right) \nonumber\\
	&=\left(\frac{\tau_g (n_l-\tau_l)}{\tau_l} \right)^2 - \left(\frac{n_l \tau_g}{\tau_l} \right)^2 +\frac{d n_l \tau_g}{\tau_l} \nonumber\\
	&=\tau_g^2\frac{-2 n_l +\tau_l}{\tau_l} +\tau_g \frac{d n_l}{\tau_l} \label{eq:tausigextlem}
	\end{align}
	also fulfills~\eqref{eq:tausiglem}.
	From the derivative in~$\tau_g$ and~$\frac{d n_l}{\tau_l} > 0$, it follows that the inequality holds for all values between the two roots~$\tau_{1,2}$ of this function in~$\tau_g$. The roots are~$\tau_1 = 0$ and 
	\begin{align*}
	\tau_2 &= \frac{n_l d}{2n_l-\tau_l} = \frac{d}{2 - \frac{\tau_l}{n_l}} = \frac{d \tau_l}{2\tau_l-\frac{\tau_l^2}{n_l}}\\
	&\stackrel{(a)}{=} \frac{d \tau_l}{2(n_l-\sqrt{n_l(n_l-\rho)})-\left(\frac{n_l^2 - 2n_l\sqrt{n_l(n_l-\rho)} +n_l(n_l-\rho)}{n_l}\right)} \\
	&= \frac{d}{\rho} \cdot \tau_l,
	\end{align*}
	where $(a)$ follows from replacing $\tau_l$ in the denominator by the Johnson radius for the local code.
	Since~\eqref{eq:tausigextlem} only holds if~$\tau_g\geq d/2$, the decoding radius is
	\begin{align}
	\frac{d}{2} \leq \tau_g &< \frac{d}{\rho} \cdot\tau_l .
	\end{align}
	There are at most~$\binom{\frac{n}{n_l}}{\sigma}$ choices of~$\sigma$ list decodable repair sets and each of these choices gives at most~$\ell_{(n_l,\rho,\tau_l)}^\sigma$ distinct possibilities to shorten the received word. The list size of each shortened code is upper bounded by~$\ell_{(n-\sigma n_l,d,\tau_g)}$ and the upper bound on the global maximum list $\ell_g$ follows.
\end{IEEEproof}

\begin{example}
Let~$\code$ be an~$\LRC{63}{16}{8}{14}$ optimal locally repairable code achieving~\eqref{eq:singletonlrc} with equality, i.e.,~$d=35$. It follows that BMD decoding corrects for up to~$t_{\text{BMD}} = 17$ errors uniquely and with~\eqref{eq:johnsonradius} we get a list decoding radius of~$\tau < 21$, i.e.,~$t = 20$. Using the principle from~\theoref{theo:ListDecodingLRCs}, with~\eqref{eq:jblrc} we obtain~$\tau_g < 22.18$, i.e.,~$t_g = 22$. Hence, two additional errors can be corrected.
\end{example}

\subsection{List Decoding Algorithm} \label{subsec:listdecodingalgo}

To achieve the decoding radius of \theoref{theo:ListDecodingLRCs}, several steps have to be taken sequentially, as shown in \algoref{algo:fanta}. While \lemref{lemma:sigma} guarantees that at least~$\sigma$ repair sets can be decoded, it does not guarantee that all repair sets for which the local decoder is able to return a local codeword are decoded correctly. For this reason, all combinations of seemingly correct local repair sets have to be tried in order to guarantee finding the correct one. 

\algoref{algo:fanta} can be improved in terms of complexity, e.g., by considering the number of errors corrected in the local codes and decreasing the decoding radius of the shortened code accordingly. However, as this is not the focus of this work, such performance optimizations are not considered here.

\begin{algorithm}
	
	\KwData{Received word~$w=c+e$ with~$c\in \LRC{n}{k}{r}{\rho}$}
	\KwResult{List of codewords within radius~$\tau_g$ of~$w$}
	
	\ForEach{Local code}{
		Decode up to~$\tau_l$~$\Rightarrow$~$\xi \geq \sigma$ repair sets with~$\ell_l \geq 1$}
	\ForEach{of the~$\binom{\xi}{\sigma}$ combinations of local repair sets with~$\ell_l \geq 1$ \label{step5}}{
		\ForEach{combination of codewords in the current~$\sigma$ local lists \label{step6}}{
			Shorten~$w$ and decode as~$(n-\sigma n_l,d)$ code up to radius~$\tau_g$} 
		}
		
		Return all codewords~$c$ with~$\dt{w}{c} \leq t_g$

	\caption{List Decoder} \label{algo:fanta}
\end{algorithm}

\algoref{algo:fanta} gives a description of the decoding steps. Its complexity is polynomial in $n$ when the number of repair sets~$\frac{n}{n_l}$ is constant, as~$\xi= O(n^{\frac{n}{n_l}})$ grows exponentially otherwise.

\subsection{Probabilistic Unique Decoder} \label{subsec:probdec}
Even for a moderate number of local repair sets, the worst case complexity of \algoref{algo:fanta} can be rather high. In \stepref{step5} all combinations of corrected local repair sets have to be tried because an undetected error event might occur, i.e., a local code might return a list with~$\ell_l >0$ that does not contain the correct codeword. Further, in \stepref{step6}, all combinations of the codewords in the local lists have to be tried to guarantee finding one that consists only of correct local codewords. It follows that whether these steps are required depends on the probability of the local list size being larger than one and on the probability of a local list with~$\ell>0$ not containing the correct local codeword. 

We can define a \emph{probabilistic unique decoder} by requiring that all local decoders return a list of size one. 
The decoding complexity is then reduced to performing the local decoding steps, shortening, and global list decoding only once.
\begin{theorem}[Probabilistic Decoding]
	An~$\LRC{n}{k}{r}{\rho}$ can be uniquely decoded up to radius~$\tau_g$ of~\eqref{eq:jblrc} with success probability
	\begin{equation}
		P_{suc} \geq P_{F} Pr\{\ell_{(n_l,\rho,\tau_l)} = 1 \}^\sigma Pr\{\ell_{(n-\sigma n_l,d,\tau_g)} = 1\}, \label{eq:sucprob}
	\end{equation}
	where $P_{F} = \sum_{i=0}^{\frac{n}{n_l}-\sigma} P_E^i (1-P_E )^{\frac{n}{n_l}-\sigma -i} \frac{1}{\binom{\sigma+i}{\sigma}}$ and $P_E$ denotes the maximum probability that a local codeword is within distance~$t_l$ for any number of errors~$>t_l$.
\end{theorem}
\begin{IEEEproof}
	$P_F$ bounds the probability that no repair set with an undetected error event is one of the~$\sigma$ repair sets which are removed by shortening. By \lemref{lemma:sigma}, undetected error events can occur for at most~$\frac{n}{n_l} -\sigma$ repair sets.~$P_F$ sums over the number of possible undetected error events and weights the probability of that number of undetected error events happening, with the probability of choosing none of them for the~$\sigma$ shortened repair sets. If no undetected error events are within the~$\sigma$ repair sets used for shortening, the result is unique if all list decoders have a list size of~$1$, giving \eqref{eq:sucprob}.
\end{IEEEproof}
For the probabilistic decoder the constraint of~$\frac{n}{n_l} = const.$ can be lifted, as its complexity grows only linear with the number of local repair sets. Further, we note that the bound on the success probability is pessimistic, as it does not consider the distribution of errors, but only the probability of success for the worst case distribution of errors. However, even this bound is close to~$1$ in many cases (see \secref{subsec:ProbTB}).

\subsection{Asymptotic Behavior}
When considering codes without locality, the asymptotic behavior is usually characterized by regarding the normalized decoding radius over the normalized distance. For codes with locality the distance depends not only on the length and dimension, but also the locality~$r$ and local distance~$\rho$, which yields different views on the asymptotic behavior.\\ 
Consider an optimal~$\LRC{n}{k}{r}{\rho}$ code with~$r|k$ and~${(r+\rho-1)|n}$. By~\eqref{eq:singletonlrc} the code rate is given by:
\begin{equation}
	R = \left(1-\frac{d}{n} + \frac{\rho}{n} \right) \frac{r}{r+\rho-1} 
	= \left(1-\frac{d}{n} + \frac{\rho}{n_l} \frac{n_l}{n} \right) R_l , \label{eq:rate}
\end{equation}
where~$R_l$ denotes the rate of the local codes. It follows that the rate~$R$ only depends on the normalized distance~$\frac{d}{n}$, the local normalized distance~$\frac{\rho}{n_l}$, the number of repair sets~$\frac{n}{n_l}$ and the local rate~$R_l$.

In this paper, we scale~$r$ and~$\rho$ such that the number of local repair sets~$\frac{n}{n_l}$ is constant, as well as the ratio~$\frac{d}{\rho}$ between local and global distance.
\figref{fig:cw_asympt} gives a graphical illustration of this scaling, where~$(a)$ depicts a short codeword and~$(b)$ and~$(c)$ depict codewords of longer codes. Note that, as indicated by the marked redundancy, the short code has the same normalized distance as the other two. The difference between~$(b)$ and~$(c)$ is due to the scaling of the parameters, where for~$(b)$ the local distance and the repair set size are the same as in $(a)$ while for~$(c)$ both scale with~$n$. 
We are interested in the latter, which can be interpreted in several ways, e.g. assume each repair set corresponds to a data center and the codeword symbols are distributed over several servers. Adapting the code to an increasing number of servers in each data center corresponds to increasing the size of each repair set while keeping the normalized distance (local storage overhead) constant. Thus, we characterize LRCs asymptotically by a fixed relation~$\beta = \frac{n \rho}{n_l d}$ between the normalized local and global distance. 
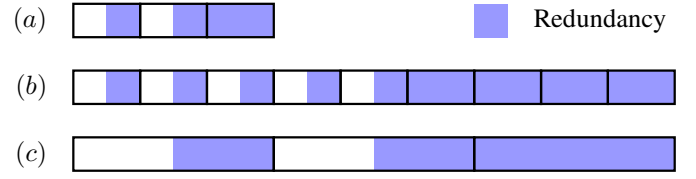
\begin{figure}

\def\x{0.05\columnwidth} 

\begin{tikzpicture}

	\node[draw=none,anchor = east] at (\x*-0.5,\x*4.5)  {$(a)$};
	\node[draw=none,anchor = east] at (\x*-0.5,\x*2.5)  {$(b)$};
	\node[draw=none,anchor = east] at (\x*-0.5,\x*0.5)  {$(c)$};
	
	\fill[blue!40!white] (\x*12,\x*5) rectangle (\x*13,\x*4);
	\node[draw=none,anchor = west] at (\x*13.5,\x*4.5)  {Redundancy};
	
	\fill[blue!40!white] (\x*1,\x*5) rectangle (\x*2,\x*4);
	\fill[blue!40!white] (\x*3,\x*5) rectangle (\x*4,\x*4);
	\fill[blue!40!white] (\x*4,\x*5) rectangle (\x*6,\x*4);
	
	\draw[line width = 0.3mm] (\x*0,\x*5) rectangle (\x*2,\x*4);
	\draw[line width = 0.3mm] (\x*2,\x*5) rectangle (\x*4,\x*4);
	\draw[line width = 0.3mm] (\x*4,\x*5) rectangle (\x*6,\x*4);

	\fill[blue!40!white] (\x*1,\x*3) rectangle (\x*2,\x*2);
	\fill[blue!40!white] (\x*3,\x*3) rectangle (\x*4,\x*2);
	\fill[blue!40!white] (\x*5,\x*3) rectangle (\x*6,\x*2);
	\fill[blue!40!white] (\x*7,\x*3) rectangle (\x*8,\x*2);
	\fill[blue!40!white] (\x*9,\x*3) rectangle (\x*10,\x*2);
	\fill[blue!40!white] (\x*10,\x*3) rectangle (\x*18,\x*2);
	
	\draw[line width = 0.3mm] (\x*0,\x*3) rectangle (\x*2,\x*2);
	\draw[line width = 0.3mm] (\x*2,\x*3) rectangle (\x*4,\x*2);
	\draw[line width = 0.3mm] (\x*4,\x*3) rectangle (\x*6,\x*2);
	\draw[line width = 0.3mm] (\x*6,\x*3) rectangle (\x*8,\x*2);
	\draw[line width = 0.3mm] (\x*8,\x*3) rectangle (\x*10,\x*2);
	\draw[line width = 0.3mm] (\x*10,\x*3) rectangle (\x*12,\x*2);
	\draw[line width = 0.3mm] (\x*12,\x*3) rectangle (\x*14,\x*2);
	\draw[line width = 0.3mm] (\x*14,\x*3) rectangle (\x*16,\x*2);
	\draw[line width = 0.3mm] (\x*16,\x*3) rectangle (\x*18,\x*2);

	\fill[blue!40!white] (\x*3,\x*1) rectangle (\x*6,\x*0);
	\fill[blue!40!white] (\x*9,\x*1) rectangle (\x*12,\x*0);
	\fill[blue!40!white] (\x*12,\x*1) rectangle (\x*18,\x*0);
	
	\draw[line width = 0.3mm] (\x*0,\x*1) rectangle (\x*6,\x*0);
	\draw[line width = 0.3mm] (\x*6,\x*1) rectangle (\x*12,\x*0);
	\draw[line width = 0.3mm] (\x*12,\x*1) rectangle (\x*18,\x*0);
	
\end{tikzpicture}
	\vspace{-15pt}
	\caption{Illustration of asymptotic scaling of parameters}
	\vspace{-10pt}
	\label{fig:cw_asympt}
\end{figure}

To compare our list decoding radius~\eqref{eq:jblrc} with the Johnson radius~\eqref{eq:johnsonradius}, rewrite
\begin{equation}
\frac{\rho}{n_l} = \beta \cdot \frac{d}{n} . \label{eq:rhonl3}
\end{equation}
For the normalized increased decoding radius it holds that
\begin{align}
\frac{\tau_g}{n} &= \frac{d\tau_l}{n \rho} = \frac{d}{n}\frac{n_l-\sqrt{n_l(n_l-\rho)}}{\rho} \nonumber\\
&= \frac{d}{n} \frac{n_l}{\rho} \frac{1-(1-\frac{\rho}{n_l})}{1+\sqrt{1-\frac{\rho}{n_l}}} = \frac{d}{n} \frac{1}{1+\sqrt{1-\frac{\rho}{n_l}}} \nonumber \\
&= \frac{d}{n} \frac{1}{1+\sqrt{1-\beta \frac{d}{n}}}, \quad\text{with} \; \beta \cdot \frac{d}{n} \leq 1. \label{eq:asympnorm}
\end{align}
Thus, the normalized decoding radius of the global code given by~\eqref{eq:asympnorm} depends only on the normalized distance of the code and the normalized decoding radius of the local code.
In \figref{fig:asympplot} the normalized decoding radii are compared for different values of~$\beta$. When the rate of the local and the global code are equal, i.e.,~$\beta = 1$, the radius equals the alphabet-free Johnson radius~\eqref{eq:johnsonradius}. For any~$\beta>1$ our decoding radius provides a gain up to the point where~$\beta \frac{d}{n}= 1$ and the curves meets the Singleton bound.
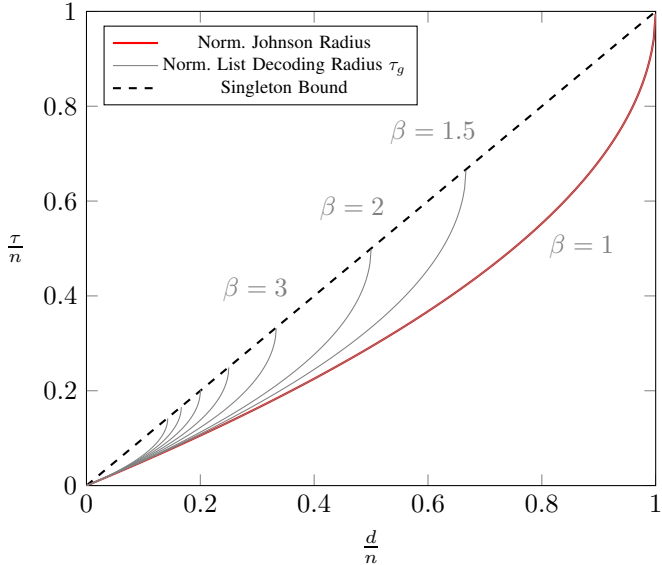
\begin{figure}[h] 
	\vspace{-10pt}
	\begin{tikzpicture}
\pgfplotsset{compat = 1.3}
\begin{axis}[
	legend style={nodes={scale=0.7, transform shape}},
	width = 1.03\columnwidth,
	xlabel = $\frac{d}{n}$,
	xlabel style = {nodes={scale=0.8, transform shape}},
	ylabel = $\frac{\tau}{n}$,
	ylabel style={rotate=-90,nodes={scale=0.85, transform shape}},
	xmin = 0,
	xmax = 1,
	ymin = 0,
	ymax = 1,
	legend pos = north west]

\addplot[color=red,
		domain = 0:1,
		samples = 300,
		thick]
		{1-sqrt(1-x)};
\addlegendentry{Norm. Johnson Radius}

\addplot[color=gray,
domain = 0:1,
samples = 300]
{x*(1/(2-(1-sqrt(1-min(1*x,1))))}
node[pos=0.66, anchor=north west] {$\beta = 1$};
\addlegendentry{Norm. List Decoding Radius $\tau_g$}

\addplot[color=black,
domain=0:1,
samples=2,
dashed,
thick]
{x};
\addlegendentry{Singleton Bound}

\addplot[color=gray,
domain = 0:1/1.5,
samples = 300]
{x*(1/(2-(1-sqrt(1-min(1.5*x,1))))};

\addplot[color=gray,
domain = 0:1/2,
samples = 300]
{x*(1/(2-(1-sqrt(1-min(2*x,1))))};

\addplot[color=gray,
domain = 0:1/3,
samples = 300]
{x*(1/(2-(1-sqrt(1-min(3*x,1))))};

\addplot[color=gray,
domain = 0:1/4,
samples = 300]
{x*(1/(2-(1-sqrt(1-min(4*x,1))))};

\addplot[color=gray,
domain = 0:0.201,
samples = 300]
{x*(1/(2-(1-sqrt(1-min(5*x,1))))};

\addplot[color=gray,
domain = 0:1/6,
samples = 300]
{x*(1/(2-(1-sqrt(1-min(6*x,1))))};

\addplot[color=gray,
domain = 0:1/7,
samples = 300]
{x*(1/(2-(1-sqrt(1-min(7*x,1))))};

\addplot[draw = none,
color=gray,
domain = 0:1,
samples = 2]
{x}
node[pos=0.7, anchor=south east] {$\beta = 1.5$}
node[pos=0.54, anchor=south east] {$\beta = 2$}
node[pos=0.37, anchor=south east] {$\beta = 3$};

\end{axis}
\end{tikzpicture} \vspace{-20pt}
	\caption{Normalized list decoding radius~$\tau_g$, see~\eqref{eq:jblrc}, with local decoding up to the Johnson radius and~$\beta = \frac{n \rho}{d n_l}$, compared to the normalized global Johnson radius}
	\vspace{-15pt}
	\label{fig:asympplot}
\end{figure}

\section{Decoder for Tamo-Barg LRCs} \label{sec:tamobarg}

\algoref{algo:fanta} provides a decoding procedure up to the radius of \eqref{eq:jblrc}. To be feasible, it requires an efficient list decoding algorithm of the global and local code, as well as an efficient way to shorten the code by known positions. While shortening is a commonly used way to decrease the length of a code, it is usually done at the encoder, where it suffices to set information symbols to zero. To shorten a code by some known positions at the decoder, all codewords that differ in the known positions can be removed from the codebook. While this gives a code of desired distance and dimension, the structure of the code is lost and it is unclear how to decode in this newly obtained code. This section addresses this problem for RS codes and shows how to efficiently apply \algoref{algo:fanta} to list decoding the Singleton-optimal RS-like codes by Tamo and Barg~\cite{Tamo2014}.

\subsection{Tamo--Barg Family of Optimal LRCs}

An $\RS{n}{d}$ Reed--Solomon code of length $n$ and distance $d$ over a field $\F{q}$ is defined as the evaluation of all polynomials $f(x) \in \F{q}[x]$ of degree $\leq k-1$ in a set $\mathcal{A} = \{\alpha_0, \alpha_1,..., \alpha_{n-1}\}$ of $n \leq q$ distinct elements of $\F{q}$. It is well known that RS codes are \emph{maximum distance separable} (MDS), i.e., have a distance of $d=n-k+1$. 

In \cite{Tamo2014} a new family of LRCs was introduced, which achieves the Singleton-like bound \eqref{eq:singletonlrc} on the distance for codes with locality and can also be defined by polynomial evaluation.
\begin{definition}[Tamo--Barg LRCs, {\cite[Constr. 8]{Tamo2014}}] \label{def:lrc8}
	Let there be a partition $\mathcal{H} = \left[H_1,...,H_{\frac{n}{n_l}}\right]$ with $|H_i| = r+\rho-1$ of a set $A \subset \F{q}$ with $|A| = n$ and a polynomial $g(x)$ of degree $r+\rho-1$ for which $g(\alpha_j) = \beta_i \; , \; \forall \; \alpha_j \in H_i$. The $\LRC{n}{k}{r}{\rho}$ code is given by the evaluation polynomial
\begin{equation}
f_u(x) = \sum_{\substack{i=0 \\ i \!\!\mod (r+\rho-1) = 0,...,r-1}}^{k-1+\left(\frac{k}{r}-1\right)(\rho-1)} \hspace{-20pt} u_i g(x)^{\floor{\frac{i}{r+\rho-1}}} x^{i\!\!\mod (r+\rho-1)} \label{eq:lrc8}
\end{equation}
and the evaluation map
\begin{align*}
	&\F{q}^k \rightarrow \F{q}^n \\
	&\uv{} \mapsto \text{ev}(\uv{}) = [f_u(\alpha_0),f_u(\alpha_1),...,f_u(\alpha_{n-1})] .
\end{align*}
\end{definition}
The polynomial in \eqref{eq:lrc8} fulfills $\deg(f_u(x)) \leq k-1+\left(\ceil{\frac{k}{r}}-1\right)(\rho-1)$ and it follows that $\LRC{n}{k}{r}{\rho} \subseteq \RS{n}{k+\left(\ceil{\frac{k}{r}}-1\right)(\rho-1)}$. We refer to this RS code containing the LRC as its \emph{supercode}. It follows that $\code$ can be decoded globally as an RS code, a well-known class of codes for which a large number of decoders exist, including the Guruswami--Sudan list decoder \cite{Guruswami1999}, which can decode errors up to the Johnson radius. Further, each local repair set is an $\RS{r}{\rho}$ code with a linear combination of the entries of $\mathbf{u}$ as message and can therefore also be efficiently (list-) decoded up to the local Johnson radius.

\subsection{List decoding Tamo-Barg LRCs}

\algoref{algo:fanta} consists of three major steps: decoding locally, shortening the code, and decoding the shortened code. As the local codes of \defref{def:lrc8} are RS codes, we can list decode the $[r+\rho-1,r,\rho]$ local codes up to the Johnson radius~\eqref{eq:johnsonradius}.
For shortening, denote the number of positions in a word $w=c+e$ with $c\in \code$ that are known to be free of error by $\delta$. The $[n,k,d]$ code $\code$ can be shortened by removing all codewords from the codebook that differ from $w$ in these positions. The obtained code is an $(n-\delta,k-\delta,d)$ code which is in general non-linear. Further, the structure of the shortened code is generally unknown, making efficient decoding difficult. To obtain a \emph{linear and structured} shortened code, we give a bijective map from the $(n-\delta,k-\delta,d)$ code to an $\RS{n-\delta}{d}$ code. For ease of notation we define the following polynomials and a corresponding set.
\begin{definition}\label{def:polred}
	For a polynomial $f(x)$ and a set $\mathcal{A} = \{\alpha_0,\alpha_1,\cdots, \alpha_{n-1}\}$, define:
	\begin{align*}
	f^{(0)}(x) &= \frac{f(x)- f(\alpha_{0})}{(x-\alpha_{0})},
	f^{(s)}(x) = \frac{f^{(s-1)}(x)-f^{(s-1)}(\alpha_{s})}{x-\alpha_{s}},
	\end{align*}
	and $\mathcal{A}^{(s)} = \{\alpha_{s+1},\alpha_{s+2},\cdots,\alpha_{n-1}\}$.
\end{definition}

\begin{lemma} \label{lemma:RSred}
	Let $\RS{n}{d}$ be the evaluation of all polynomials $f(x)$ with $\deg(f(x)) \leq k-1$ on the set $\mathcal{A} = \{\alpha_0,\cdots, \alpha_{n-1}\}$. Then the evaluation of all corresponding polynomials $f^{(\delta-1)}(x)$ on $\mathcal{A}^{(\delta-1)}$ is an $\RS{n-\delta}{d}$ code.
\end{lemma}
\begin{proof}
	It needs to be shown, that for any $f(x)$ with $\deg(f(x)) \leq k-1$ it holds that ${\deg(f^{(i)}(x)) \leq k-1-i}$. The polynomial $f'(x) = f(x) - f(\alpha_0)$ has a root at $\alpha_0$ and hence $f^{(0)}(x)$ with $f'(x) = f^{(0)}(x) (x-\alpha_0)$ exists. It follows that $ \deg(f^{(0)}(x)) = \deg(f'(x))-1 = \deg(f(x)) -1 \leq k-2$. The generalization to $f^{(i)}(x)$ follows by induction.
\end{proof}
Since most positions in a codeword are free of error, it makes sense to define a relation between the error vector of the shortened code and the original code.
\begin{lemma}
	Let $g_i(x) = f(x)+e_i$ with $e_i = 0, \; \forall \; i<\delta$. Then 
	\begin{align*}
	g^{(\delta-1)}_i(x) &= f^{(\delta-1)}(x) + e^{(\delta-1)}_i		
	\end{align*}
	with
	\begin{equation} \label{eq:errred}
	e_i^{(\delta-1)} = e_i \prod_{j=0}^{\delta-1} (x-\alpha_j)^{-1}
	\end{equation}
\end{lemma}
\begin{proof}
	For $\delta>0$, applying \defref{def:polred} gives
	\begin{align*}
	g_i^{(0)}(x) &= ((f(x)+e_i) - (f(\alpha_{0})+e_0))(x-\alpha_{0})^{-1} \\
	&= (f(x) - f(\alpha_{0}))(x-\alpha_{0})^{-1} + e_i (x-\alpha_0)^{-1} \\
	&= f^{(0)}(x) + e_i^{(0)} .
	\end{align*}
	Then, \eqref{eq:errred} follows by induction.
\end{proof}
With \lemref{lemma:RSred} and the Guruswami-Sudan decoder, all necessary tools for decoding up to radius $\tau_g$ of \eqref{eq:jblrc} are given. 

\begin{figure}
	\input{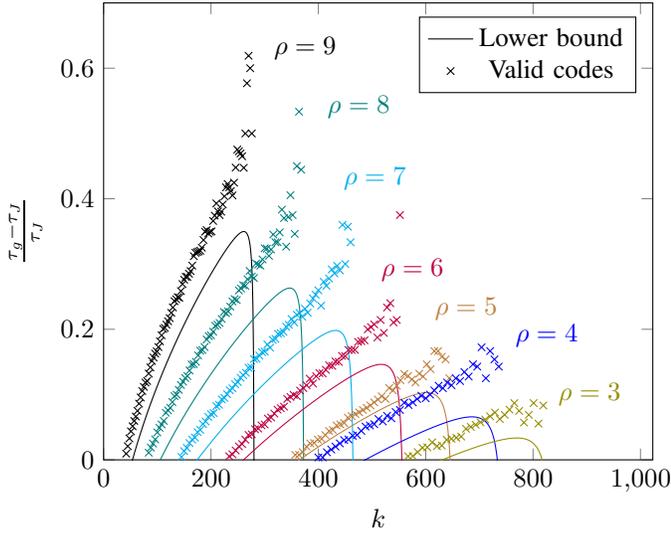}
	\vspace{-6pt}
	\caption{Relative gain in the decoding radius in relation to the Johnson radius for optimal LRCs of length $n=1023$ and repair set size $n_l=11$, where $\tau_g$ is given by \eqref{eq:jblrc} and $\tau_J$ denotes the Johnson radius \eqref{eq:johnsonradius} for the respective parameters.}
	\vspace{-10pt}
	\label{fig:relativegain}
\end{figure}

\figref{fig:relativegain} shows the relative gain for optimal LRCs of length $n=1023$ and repair set size $n_l=11$ for different values of $\rho$. For each $\rho$, a lower bound on the relative gain is given, i.e., the fraction by which our bound in \theoref{theo:ListDecodingLRCs} exceeds the Johnson radius of \eqref{eq:johnsonradius}. Each cross depicts the gain obtained for an $\LRCnp$ with $r|k$ and $(r+\rho-1)|n$, when considering the exact values for all ceiling and floor operations neglected in the derivation of~\eqref{eq:jblrc}. For example, consider the optimal $\LRC{15}{6}{3}{3}$ of distance $d=8$. Equation~\eqref{eq:jblrc} gives $\tau_g \approx 4.9$ and it follows that $4$ errors can be corrected, the same number as for the $\RS{15}{8}$ supercode. However, when considering the floor operation omitted in the proof of \lemref{lemma:sigma} by defining $\sigma' = \frac{n}{n_l} - \floor{\frac{\floor{\tau_g}}{t_l+1}}$, we see that the largest value for $\tau_g$ that fulfills 
\begin{equation}
	0 < (n-\sigma' n_l-\tau_g)^2 -(n-\sigma' n_l) (n-\sigma' n_l -d) 
\end{equation}
is $\tau_g \approx 5.52$ (the Johnson radius of an $[n=10,k=3,d=8]$ code). The gain is due to the fact that for any distribution of~$5$ errors, there will always be at least one repair set with at most $t_l=1$ errors and the code can be shortened by $n_l$ positions.

\subsection{Probabilistic Unique Decoding of Tamo-Barg Codes} \label{subsec:ProbTB}
\secref{subsec:probdec} has introduced a simple probabilistic unique decoder whose success probability depends on the likelihood of a miscorrection as well as the probability of the list sizes being equal to one. 
For RS codes, these probabilities are known to be small for a wide range of parameters \cite{McEliece1986,Cheung1988,McEliece2003}. \tabref{tab:sucprob} provides a lower bound on the success probabilities obtained by \eqref{eq:sucprob} for different $\LRCnp$ parameters. The columns labeled $\tau$ and $\tau_g$ give the bounds on the decoding radius from \eqref{eq:johnsonradius} and \eqref{eq:jblrc}, respectively. This shows that the computationally expensive case, where multiple repair sets have undetected error events and the local lists contain non-casual codewords, is highly unlikely and we can efficiently decode beyond the global Johnson radius.

\begin{table} 
	\caption{Success probabilities \eqref{eq:sucprob} of probabilistic unique decoding of~$\tau_g$ errors for different LRC parameters, where $\tau_g$ is given by \eqref{eq:jblrc} and $\tau_J$ denotes the Johnson radius \eqref{eq:johnsonradius} for the parameters.}\label{tab:sucprob}
	\vspace{-10pt}
\begin{center}
\begin{tabular}{CCCCCCCC}
n&k&r&\rho&d&\tau_J & \tau_g & P_{suc}\\ \hline
1023&99&3&9&669&421.21&469.01&94.63 \% \\
1023&129&3&9&559&334.03 &391.89 & 94.00 \% \\
1023&220&5&7&546&324.45&340.60&97.15 \% \\
1023&250&5&7&480&277.68&299.43&97.16 \% \\
1023&390&6&6&314&171.35&187.55&99.03 \% \\
1023&420&6&6&259&138.93&154.70&99.11 \% \\
1023&560&7&5&148&76.89&85.12&99.85 \% \\
1023&590&7&5&98&50.23&56.36&99.90 \%
\end{tabular}
\end{center}
\vspace{-14pt}
\end{table}
\vspace{-3pt}

\bibliographystyle{IEEEtran}
\bibliography{IEEEabrv,bib_lrctb}

\begin{thebibliography}{10}
\providecommand{\url}[1]{#1}
\csname url@samestyle\endcsname
\providecommand{\newblock}{\relax}
\providecommand{\bibinfo}[2]{#2}
\providecommand{\BIBentrySTDinterwordspacing}{\spaceskip=0pt\relax}
\providecommand{\BIBentryALTinterwordstretchfactor}{4}
\providecommand{\BIBentryALTinterwordspacing}{\spaceskip=\fontdimen2\font plus
\BIBentryALTinterwordstretchfactor\fontdimen3\font minus
  \fontdimen4\font\relax}
\providecommand{\BIBforeignlanguage}[2]{{%
\expandafter\ifx\csname l@#1\endcsname\relax
\typeout{** WARNING: IEEEtran.bst: No hyphenation pattern has been}%
\typeout{** loaded for the language `#1'. Using the pattern for}%
\typeout{** the default language instead.}%
\else
\language=\csname l@#1\endcsname
\fi
#2}}
\providecommand{\BIBdecl}{\relax}
\BIBdecl

\bibitem{Gopalan2011}
P.~Gopalan, C.~Huang, H.~Simitci, and S.~Yekhanin, ``{On the locality of
  codeword symbols},'' \emph{IEEE Trans. Inf. Theory}, vol.~58, no.~11, pp.
  6925--6934, Nov. 2012.

\bibitem{Kamath2012}
G.~M. Kamath, N.~Prakash, V.~Lalitha, and P.~V. Kumar, ``{Codes with local
  regeneration},'' in \emph{IEEE Int. Symp. Inf. Theory}, Jul. 2013, pp.
  1606--1610.

\bibitem{Silberstein2013}
N.~Silberstein, A.~S. Rawat, O.~O. Koyluoglu, and S.~Vishwanath, ``{Optimal
  locally repairable codes via rank-metric codes},'' \emph{IEEE Int. Symp. Inf.
  Theory}, pp. 1819--1823, 2013.

\bibitem{Tamo2013}
I.~Tamo, D.~S. Papailiopoulos, and A.~G. Dimakis, ``{Optimal locally repairable
  codes and connections to matroid theory},'' \emph{IEEE Int. Symp. Inf.
  Theory}, pp. 1814--1818, Jul. 2013.

\bibitem{Tamo2014}
I.~Tamo and A.~Barg, ``{A family of optimal locally recoverable codes},''
  \emph{IEEE Trans. Inf. Theory}, vol.~60, no.~8, pp. 4661--4676, Aug. 2014.

\bibitem{Johnson1962}
S.~Johnson, ``{A new upper bound for error-correcting codes},'' \emph{IEEE
  Trans. Inf. Theory}, vol.~8, no.~3, pp. 203--207, Apr. 1962.

\bibitem{Guruswami1999}
V.~Guruswami and M.~Sudan, ``{Improved decoding of Reed-Solomon and
  algebraic-geometry codes},'' \emph{IEEE Trans. Inf. Theory}, vol.~45, no.~6,
  pp. 1757--1767, Sep. 1999.

\bibitem{Rudra2013}
A.~Rudra and M.~Wootters, ``Every list-decodable code for high noise has
  abundant near-optimal rate puncturings,'' in \emph{Proc. Forty-sixth Annual
  ACM Symp. on Theory of Computing}, New York, NY, USA, 2014, pp. 764--773.

\bibitem{Guruswami2006}
V.~Guruswami, ``{Algorithmic results in list decoding},'' \emph{Found. Trends
  Theor. Comput. Sci.}, vol.~2, no.~2, pp. 107--195, 2006.

\bibitem{Guruswami2012}
V.~Guruswami and C.~Xing, ``{List decoding Reed-Solomon, algebraic-geometric,
  and Gabidulin subcodes up to the Singleton bound},'' \emph{Electronic Colloq.
  on Computational Complexity}, no. 146, 2012.

\bibitem{McEliece1986}
{McEliece, R. J. and Swanson, L.}, ``{On the decoder error probability for Reed
  - Solomon codes},'' \emph{IEEE Trans. Inf. Theory}, vol.~32, no.~5, pp.
  701--703, Sep. 1986.

\bibitem{Cheung1988}
K.-M. Cheung and R.~J. McEliece, ``The undetected error probability for
  {R}eed-{S}olomon codes,'' in \emph{Military Comm. Conf.}, vol.~1, Oct. 1988,
  pp. 163--167.

\bibitem{McEliece2003}
R.~J. McEliece, ``{The Guruswami–Sudan decoding algorithm for Reed–Solomon
  codes},'' \emph{IPN Progress Report}, vol. 42-153, 2003.

\end{thebibliography}

%
%
%
%
%
%
%

\end{document}